\documentclass[a4paper, 10pt, twosides]{amsart}

\usepackage{lmodern}
\usepackage{amssymb}
\usepackage{amsthm}
\usepackage{amsaddr}
\usepackage{mathtools}
\usepackage{MnSymbol}
\usepackage{graphicx}
\usepackage{caption}
\usepackage{subcaption}
\usepackage{url}

\usepackage{enumitem}
\setitemize{nolistsep}
\setenumerate{nolistsep}
\usepackage{todonotes}
\allowdisplaybreaks

\newtheorem{theorem}{Theorem}

\theoremstyle{definition}
\newtheorem{defn}{Definition}
\newtheorem{example}{Example}

\newtheorem{problem}{Problem}

\theoremstyle{remark}
\newtheorem{rem}{Remark}

\theoremstyle{assumption}
\newtheorem{assump}{Assumption}

\theoremstyle{fact}

\theoremstyle{claim}

\numberwithin{equation}{section}

\newcommand{\norm}[1]{\left\lVert{#1}\right\rVert}
\newcommand{\abs}[1]{\left\lvert{#1}\right\rvert}

\newcommand{\Let}{\coloneqq}

\newcommand{\pmat}[1]{\begin{pmatrix}#1\end{pmatrix}}

\renewcommand{\geq}{\geqslant}

\renewcommand{\leq}{\leqslant}

\newcommand{\R}{\mathbb{R}}
\newcommand{\N}{\mathbb{N}}
\renewcommand{\P}{\mathcal{P}}

\newcommand{\Dm}{\delta}
\newcommand{\DM}{\Delta}

\title[Stability under restricted switching]{Robust matrix commutator conditions for stability of switched linear systems under\\restricted switching}
\author{Atreyee Kundu}
\address{Department of Electrical Engineering,\\Indian Institute of Science Bangalore,\\Bengaluru - 560012, India,\\ E-mail: atreyeek@iisc.ac.in,\\Webpage: www.ee.iisc.ac.in/faculty/atreyee}
\author{Debasish Chatterjee}
\address{Systems \& Control Engineering,\\Indian Institute of Technology Bombay,\\Mumbai - 400076, India,\\E-mail: dchatter@iitb.ac.in,\\Webpage: {www.sc.iitb.ac.in/\~{}chatterjee}}


\date{\today}

\begin{document}

	\begin{abstract}
        This article treats global uniform exponential stability (GUES) of discrete-time switched linear systems under restricted switching. Given admissible minimum and maximum dwell times, we provide sufficient conditions on the subsystems under which they admit a set of switching signals that obeys the given restrictions on dwell times \emph{and} preserves stability of the resulting switched system. Our analysis relies on combinatorial arguments applied to matrix commutators and avoids the employment of Lyapunov-like functions. The proposed set of stabilizing switching signals is characterized in terms of duration of activation of Schur stable subsystems and non-consecutive activation of distinct unstable subsystems.
    \end{abstract}

    \maketitle
\section{Introduction}
\label{s:intro}
    Hybrid systems find wide applications in modern day Cyber-Physical Systems (CPS). In this article we deal with an abstraction of hybrid systems, where we focus on the discrete dynamics and abstract away the continuous dynamics as switching. Such an abstraction is called a discrete-time \emph{switched system}, and contains two ingredients --- a family of systems and a switching signal. The \emph{switching signal} selects an \emph{active subsystem} at every instant of time, i.e., the system from the family that is currently being followed \cite[\S 1.1.2]{Liberzon}. A vast body of hybrid systems literature is devoted to stability of switched systems.

    Given a family of systems, much attention has been devoted to finding estimates of minimum dwell times on stable subsystems and maximum dwell times on unstable subsystems under which stability of a switched system is preserved; see e.g., \cite{ghi,Liberzon'12,Zhai'02}. On the one hand, these stability conditions are only sufficient and do not imply instability under a set of dwell times different from the estimated ones. On the other hand, the minimum and maximum dwell times on subsystems are often governed by physical constraints of a system. For instance, actuator saturations may prevent switching frequencies beyond a certain limit, or in order to switch from one component to another, a system may undergo certain operations of non-negligible durations leading to a minimum dwell time requirement on each subsystem. In addition, systems whose components need regular maintenance or replacements, e.g., aircraft carriers, MEMS systems, etc. and systems that are dependent on diurnal or seasonal changes, e.g., components of an electricity grid have inherent restrictions on admissible maximum dwell times \cite{Heydari'17, abc}. It is, therefore, of interest to study properties of a switched system under pre-specified restrictions on dwell times. Stability and optimal control of switched systems under restricted dwell times were dealt with earlier in the literature, see e.g., \cite{Hernandez-Vargas'12,Heydari'17,Jungers'13,abc,def} and the references therein.

    In this article we consider the setting where all subsystems are linear, and study stability of a switched system under pre-specified restrictions on dwell times. Given admissible minimum and maximum dwell times, our objective is to characterize sets of subsystems such that they admit switching signals that obey the given restrictions \emph{and} preserve stability of the resulting switched system. We allow both Schur stable and unstable (not Schur stable) subsystems, and achieve our task in two steps:
     \begin{itemize}[label = \(\circ\), leftmargin = *]
        \item first, we fix a set of switching signals that obey the given restrictions, and
        \item second, we identify sufficient conditions on subsystems under which the above set of switching signals is stabilizing.
     \end{itemize}
     Our characterization of stabilizing switching signals is based on minimum dwell times on Schur stable subsystems and non-consecutive activation of distinct unstable subsystems. Sufficient conditions on the subsystem matrices are derived by employing commutation relations between certain products of these matrices.

	A switched linear system is known to be stable under arbitrary switching if the subsystem matrices are stable and commute pairwise \cite{Narendra'94} or are sufficiently ``close'' to a set of matrices whose elements commute pairwise \cite{Agrachev'12}. Recently in \cite{mno} one of the authors extended the conditions of \cite{Agrachev'12} to the setting of stability under all switching signals obeying a given minimum dwell time. The overarching assumption in the above body of results is that all subsystems are Schur stable. In this article we deal with matrix commutator based characterization of stability of switched systems under pre-specified restrictions on dwell times when not all subsystems are Schur stable. Towards this end, we follow the combinatorial analysis technique proposed in \cite{Agrachev'12}, and present two sets of sufficient conditions:
	\begin{itemize}[label = \(\circ\), leftmargin = *]
        		\item The first set of conditions caters to the setting where certain products of Schur stable and unstable subsystem matrices commute. We rely on the rate of decay of the Schur stable subsystems to guarantee stability. These conditions, however, lack robustness in the sense that if the entries of the subsystem matrices are perturbed by a margin such that the matrix products of our interest cease to commute, then our conditions are no longer useful to guarantee stability. However, stability, being a robust property, may be preserved under small perturbations in the elements of the subsystem matrices. This fact motivates our second set of stability conditions.
        		\item The second set of conditions caters to sets of subsystems for which the commutators under consideration do not necessarily vanish, but are small quantities in the induced Euclidean norm. We rely on the rate of decay of the Schur stable subsystems, upper bounds on the norms of the commutators of certain products of the subsystem matrices, and a set of scalars relating to the individual matrices and the given minimum and maximum dwell times. These conditions ensure robust stability in the sense that if perturbing the elements of the subsystem matrices does not take them ``too far'' from a set of matrices for which certain matrix products commute, then stability of a switched system remains preserved under the proposed set of switching signals.
    \end{itemize}
    We utilize the dwell time restrictions on our switching signals to split matrix products into sums, and apply counting arguments on them; see Remark \ref{rem:analysis_des} for a detailed comparison between our analysis techniques and the methods of \cite{Agrachev'12}. Since we aim for a subset of the set of all admissible switching signals for stability, we are able to accommodate unstable subsystems in the setting of matrix commutators. This is in contrast to stability under arbitrary switching tackled in \cite{Agrachev'12,Narendra'94} and minimum dwell time switching tackled in \cite{mno}, where all subsystems are necessarily Schur stable. To the best of our knowledge, this is the first instance in the literature where commutation relations between subsystem matrices are utilized to characterize stabilizing switching signals in the presence of unstable subsystems \emph{and} pre-specified restrictions on dwell times.

    The remainder of this article is organized as follows: We formulate the problem under consideration in \S\ref{s:prob_stat}, and catalog a set of preliminaries in \S\ref{s:prelims}. Our main results appear in \S\ref{s:mainres}, where we also discuss various features of our results. Numerical examples are presented in \S\ref{s:num_ex}, and we conclude in \S\ref{s:concln} with a brief discussion of open problems.

    {\bf Notation}. \(\N\) is the set of natural numbers, \(\N_{0} = \N\cup\{0\}\). \(\norm{\cdot}\) denotes the Euclidean norm (resp., induced matrix norm) of a vector (resp., a matrix). \(0_{d\times d}\) is the \(d\)-dimensional \(0\) matrix. For a matrix \(P\), given by a product of matrices \(M_{i}\)'s, \(\abs{P}\) denotes the length of the product, i.e., the number of matrices that appear in \(P\), counting repetitions.
\section{Problem statement}
\label{s:prob_stat}
    We consider a family of discrete-time linear systems
    \begin{align}
    \label{e:family}
        x(t+1) = A_{i}x(t),\:\:x(0) = x_{0},\:\:i\in\P,\:\:t\in\N_{0},
    \end{align}
    where \(x(t)\in\R^{d}\) is the vector of states at time \(t\), \(\P = \{1,2,\ldots,N\}\) is an index set, and \(A_{i}\in\R^{d\times d}\), \(i\in\P\), are known constant matrices. Let \(\sigma:\N_{0}\to\P\) be a {switching signal} that specifies at every time \(t\), the index of the active subsystem, i.e., the dynamics from \eqref{e:family} that is being followed at \(t\). A discrete-time {switched linear system} generated by the family of systems \eqref{e:family} and a switching signal \(\sigma\) is described by the recursion
    \begin{align}
    \label{e:swsys}
        x(t+1) = A_{\sigma(t)}x(t),\:\:x(0) = x_{0},\:\:t\in\N_{0}.
    \end{align}
    The solution to \eqref{e:swsys} is given by
    \begin{align}
    \label{e:traj}
    	x(t) = A_{\sigma(t-1)}A_{\sigma(t-2)}\ldots A_{\sigma(2)}A_{\sigma(1)}A_{\sigma(0)}x_{0},\:\:t\in\N,
    \end{align}
    where we have suppressed the dependence of \(x\) on \(\sigma\) for notational simplicity.
    Our focus is on global uniform exponential stability (GUES) of the switched system \eqref{e:swsys}.
    \begin{defn}{\cite[\S 2]{Agrachev'12}}
	 \label{d:gues}
	 	The switched system \eqref{e:swsys} is \emph{globally uniformly exponentially stable (GUES) over a set of switching signals \(\mathcal{S}\)} if there exist positive numbers \(c\) and \(\lambda\) such that for arbitrary choices of the initial condition \(x_{0}\) and switching signal \(\sigma\in{\mathcal{S}}\), the following inequality holds:
		\begin{align}
		\label{e:gues}
            \norm{x(t)}\leq c\exp(-\lambda t)\norm{x_{0}}\:\:\text{for all}\:\:t\in\N.
		\end{align}
	 \end{defn}
	The term `uniform' in the above definition refers to the fact that the numbers \(c\) and \(\lambda\) can be chosen irrespective of \(\sigma\). Let \(0=:\tau_{0}<\tau_{1}<\cdots\) be the points in time where \(\sigma\) ``jumps''; these are the \emph{switching instants}. In this article we will work with switching signals \(\sigma\) that satisfy the following condition: there exist \(\Dm\) and \(\DM\in\N\) such that
    \begin{align}
    \label{e:min-max_dwell}
        \Dm \leq \tau_{k+1} - \tau_{k} \leq \DM,\:k=0,1,2,\ldots.
    \end{align}
    Condition \eqref{e:min-max_dwell} implies that the duration of activation of any subsystem \(i\in\P\) is at least \(\Dm\) and at most \(\DM\) units of time. We call \(\Dm\) and \(\DM\) as the minimum and maximum dwell times, respectively. Given \(\Dm\) and \(\DM\), let \(\mathcal{S}(\Dm,\DM)\) denote the set of all switching signals \(\sigma\) that satisfy condition \eqref{e:min-max_dwell}. We will solve the following problem:
    \begin{problem}
    \label{prob:mainprob}
        Given admissible minimum and maximum dwell times \(\Dm\) and \(\DM\in\N\), \(\Dm<\DM\), find conditions on the matrices \(\{A_{i}\:|\:i\in\P\}\) such that there is a set of switching signals \(\tilde{\mathcal{S}}(\Dm,\DM)\subset\mathcal{S}(\Dm,\DM)\) over which the switched system \eqref{e:swsys} is GUES.
    \end{problem}
    \begin{rem}
    \label{rem:prob_des1}
        The classical problem of stability under dwell time switching \cite[Chapter 3]{Liberzon} deals with identifying minimum dwell time on stable subsystems and maximum dwell time on unstable subsystems such that a switched system generated by a given family of systems is stable. Formally, if \(\mathcal{S}\) is the set of all switching signals \(\sigma:\N_{0}\to\P\); we seek elements of \(\mathcal{S}\) that are stabilizing. In contrast, in Problem \ref{prob:mainprob} we consider the admissible minimum and maximum dwell times to be ``given'', and aim to identify families of systems that admit stabilizing switching signals that obey the given restrictions. In other words, we restrict our attention to the set \(\mathcal{S}(\Dm,\DM)\subset\mathcal{S}\), and find conditions on \(\{A_{i}\:|\:i\in\P\}\) such that \(\mathcal{S}(\Dm,\DM)\) contains stabilizing elements.
    \end{rem}
    \begin{rem}
    \label{rem:prob_des2}
        Recently in \cite{abc,def} one of the authors studied the algorithmic design of switching signals that preserve stability of switched nonlinear systems under pre-specified restrictions on minimum and maximum dwell times. Stabilizing switching signals were designed under the assumption that the underlying weighted directed graph of a switched system admits a certain class of cycles. While \cite{abc,def} are concerned with a ``design'' problem, in this article we are dealing with an ``existence'' problem. In particular, we restrict our attention to linear subsystems and seek for sets of subsystems that admit stabilizing switching signals under restricted dwell times.
    \end{rem}

    Towards solving Problem \ref{prob:mainprob}, we will employ two steps:
    \begin{itemize}[label = \(\circ\), leftmargin = *]
        \item first, we fix a subset \(\tilde{\mathcal{S}}(\Dm,\DM)\) of the set of switching signals \(\mathcal{S}(\Dm,\DM)\), and
        \item second, we identify sufficient conditions on the subsystem matrices \(\{A_{i}\:|\:i\in\P\}\) under which \(\tilde{\mathcal{S}}(\Dm,\DM)\) is stabilizing.
    \end{itemize}
\section{Preliminaries}
\label{s:prelims}
    Let \(\P_{S}\) and \(\P_{U}\) denote the sets of indices of Schur stable and unstable subsystems, respectively, \(\P = \P_{S}\sqcup\P_{U}\)\footnote{The set of unstable subsystems ``also'' includes Lyapunov stable but not asymptotically stable (Schur stable) subsystems.}. 
    \begin{rem}
    \label{rem:mno_des}
        The case of stability of \eqref{e:swsys} under all switching signals obeying a certain minimum dwell time is addressed recently in \cite{mno}. A necessary condition there is that all subsystem matrices \(\{A_{i}\:|\:i\in\P\}\) are Schur stable. The analysis technique presented in \cite{mno} extends readily to the setting of restricted minimum and maximum dwell times with all Schur stable subsystems, see \cite[Remark 7]{mno} for a detailed discussion. In contrast, here we focus on families of systems \eqref{e:family} that contain both Schur stable and unstable subsystems.
    \end{rem}
    We let
    \begin{align}
    \label{e:M_defn}
        M \Let \max_{i\in\P}\norm{A_{i}}.
    \end{align}
    \begin{assump}
    \label{assump:key1}
    \rm{
        There exists \(m\in\N\) with \(\Dm\leq m\leq\DM\) such that the following condition holds:
        \begin{align}
        \label{e:key_ineq1}
            \norm{A_{i}^{m}}\leq\rho < 1\:\text{for all}\:i\in\P_{S}.
        \end{align}
    }
    \end{assump}
    \begin{rem}
    \label{rem:key_assump}
        Notice that for a Schur stable matrix \(A_{i}\), there exists an integer \(m\geq 1\) such that \(\norm{A_{i}^{m}} < 1\). We will work with the smallest \(m\in\{\Dm,\Dm+1,\ldots,\DM\}\) such that for every \(n\in\{m+1,\ldots,\DM\}\), the condition \eqref{e:key_ineq1} holds with \(n=m\).\footnote{A discussion on the choice of \(m\) is provided in Remark \ref{rem:m_choice}.} Clearly, Assumption \ref{assump:key1} excludes those Schur stable matrices for which the smallest integer \(m\) satisfying \eqref{e:key_ineq1} is strictly bigger than the given integer \(\DM\).
    \end{rem}
    Let \(K_{1}\) be the largest integer satisfying \(K_{1}\Dm\leq m\) and \(K_{2}\) be the largest integer satisfying \(K_{2}\Dm\leq\DM\). We will need to employ the following relations for (matrix) commutators of products of matrices:
    \begin{align}
    \label{e:commutator_defn}
        E_{ij}^{p,q} = A_{i}^{p}A_{j}^{q} - A_{j}^{q}A_{i}^{p},\:p,q\in\{1,\Dm\},\:i\in\P_{U},\:j\in\P_{S}.
    \end{align}
    \begin{rem}
    \label{rem:comm_des}
        Notice that \(E_{ij}^{p,q}\) are commutators between products of Schur stable and unstable subsystem matrices \(A_{j}\) and \(A_{i}\) of length \(p,q\in\{1,\Dm\}\). The choice of these commutators is motivated by our set of stabilizing switching signals to be described momentarily; see Remark \ref{rem:analysis_des} for a detailed discussion.
    \end{rem}

    Given the numbers \(M, N\) and \(m\), we define the functions \(\zeta_{p,q}:\N\times\N\to\R\), \(p,q\in\{1,\Dm\}\) as follows:
    \begin{align}
    	\label{e:zeta1_defn}\zeta_{\Dm,\Dm}(\Dm,\DM) &= K_{1}K_{2}M^{(N-1)(m+\DM-1)+m+\DM-2\Dm},\\
	\label{e:zeta2_defn}\zeta_{1,\Dm}(\Dm,\DM) &= K_{1}(\DM-K_{2}\Dm)M^{(N-1)(m+\DM-1)+m+\DM-\Dm-1},\\
	\label{e:zeta3_defn}\zeta_{\Dm,1}(\Dm,\DM) &= (m-K_{1}\Dm)K_{2}M^{(N-1)(m+\DM-1)+m+\DM-\Dm-1},\\
	\label{e:zeta4_defn}\zeta_{1,1}(\Dm,\DM) &= (m-K_{1}\Dm)(\DM-K_{2}\Dm)M^{(N-1)(m+\DM-1)+m+\DM-2}.
    \end{align}
    These functions will be useful in our analysis. We are now in a position to present our results of this article.
\section{Results and discussions}
\label{s:mainres}
    Fix a switching signal \(\sigma\in\mathcal{S}(\Dm,\DM)\) that satisfies the following conditions for all \(k = 0,1,2,\ldots\):
    \begin{align}
    \label{e:sw_res1}\tau_{k+1}-\tau_{k} \geq m,&\:\text{if}\:\:\sigma(\tau_{k})\in\P_{S},\:\text{and}\\
    \label{e:sw_res2}\sigma(\tau_{k+1})\in\P_{S},&\:\text{if}\:\:\sigma(\tau_{k})\in\P_{U},
    \end{align}
    Let \(\tilde{\mathcal{S}}(\Dm,\DM)\subset\mathcal{S}(\Dm,\DM)\) denote the set of all switching signals \(\sigma\) that satisfy conditions \eqref{e:sw_res1}-\eqref{e:sw_res2}.
    \begin{rem}
    \label{rem:sw_subset}
        Given \(\Dm\) and \(\DM\), every element of \(\mathcal{S}(\Dm,\DM)\) dwells both on Schur stable and unstable subsystems for at least \(\Dm\) and at most \(\DM\) units of time. The set \(\tilde{\mathcal{S}}(\Dm,\DM)\) contains those elements of \(\mathcal{S}(\Dm,\DM)\) that dwell on Schur stable subsystems for at least \(m\geq\Dm\) units of time and do not activate two distinct unstable subsystems consecutively. Notice that if \(m=\DM\), then \(\tau_{k+1}-\tau_{k} = \DM\) for \(\sigma(\tau_{k})\in\P_{S}\), \(k=0,1,2,\ldots\).
    \end{rem}
    \begin{rem}
    \label{rem:all_stable_unstable}
    	If \(\P_{U} = \emptyset\), then the elements of \(\tilde{\mathcal{S}}(\Dm,\DM)\) are the ones that obey a minimum dwell time \(m\) and a maximum dwell time \(\DM\) on every subsystem. In view of our choice of \(m\) described in Remark \ref{rem:key_assump}, the norms of each product \(A_{j}^{\tau_{k+1}-\tau_{k}}\), \(k=0,1,\ldots\), \(j\in\P_{S}\) is strictly less than \(1\). Consequently, \(\norm{A_{\sigma(t-1)}A_{\sigma(t-2)}\cdots A_{\sigma(1)}A_{\sigma(0)}}\to 0\) as \(t\to+\infty\). Notice that the set of switching signals \(\tilde{\mathcal{S}}(\Dm,\DM)\) is not defined if \(\P_{S} = \emptyset\). Indeed, otherwise condition \eqref{e:sw_res2} is violated.
    \end{rem}

    \begin{rem}
    \label{rem:unstab_activation}
        The elements of \(\tilde{\mathcal{S}}(\Dm,\DM)\) restricts consecutive activation of distinct unstable subsystems. Theoretically, this feature is restrictive. However, in many practical contexts, it is a natural choice for a switching mechanism to take a system from a faulty component to a healthy component. The elements of \(\tilde{\mathcal{S}}(\Dm,\DM)\) cater to this setting.
    \end{rem}

    When both \(\P_{S}\) and \(\P_{U} \neq \emptyset\), stability of \eqref{e:swsys} under an element \(\sigma\) of \(\tilde{\mathcal{S}}(\Dm,\DM)\) depends on the choice of the subsystem matrices \(\{A_{i}\:|\:i\in\P\}\). We demonstrate this fact in the following example:
    \begin{example}
    \label{ex:counter_ex1}
        Consider \(\P = \{1,2\}\) with
        \[
            A_{1} = \pmat{-0.24 & 0.14\\-0.85 & -0.89}\:\:\text{and}\:\:A_{2} = \pmat{0.12 & 1.12\\1.74 & -1.48}.
        \]
        Clearly, \(\P_{S} = \{1\}\) and \(\P_{U} = \{2\}\). Let \(\Dm = 2\) and \(\DM = 3\). We have
        \[
            \norm{A_{1}^{2}} = 1.18\:\:\text{and}\:\:\norm{A_{1}^{3}} = 0.95.
        \]
        Consequently, \(m=3\). Let a switching signal \(\sigma\) satisfy
        \[
            \tau_{i+1} - \tau_{i} = m = 3,\:i=0,1,2,\ldots.
        \]
        Clearly, \(\sigma\in\tilde{\mathcal{S}}(2,3)\). We observe that the switched system \eqref{e:swsys} is unstable under the above \(\sigma\). In Figure \ref{fig:counter_xplot1} we illustrate the corresponding \((\norm{x(t)})_{t\in\N_{0}}\). The initial condition for this plot is chosen as \(x_{0} = \pmat{-1\\1}\).

        Now, consider
        \[
            \tilde{A}_{{1}} = A_{1}\:\:\text{and}\:\:\tilde{A}_{2} = \pmat{0.10 & 0.90\\0.50 & -1.20}.
        \]
        It is observed that the switching signal \(\sigma\) under consideration, is stabilizing. The corresponding plot of \((\norm{x(t)})_{t\in\N_{0}}\) with initial condition \(x_{0} = \pmat{-1\\1}\), is shown in Figure \ref{fig:counter_xplot2}.
        \begin{figure}[htbp]
	    \centering
		\begin{subfigure}{.5\textwidth}
  		\centering
  			\includegraphics[scale = 0.3]{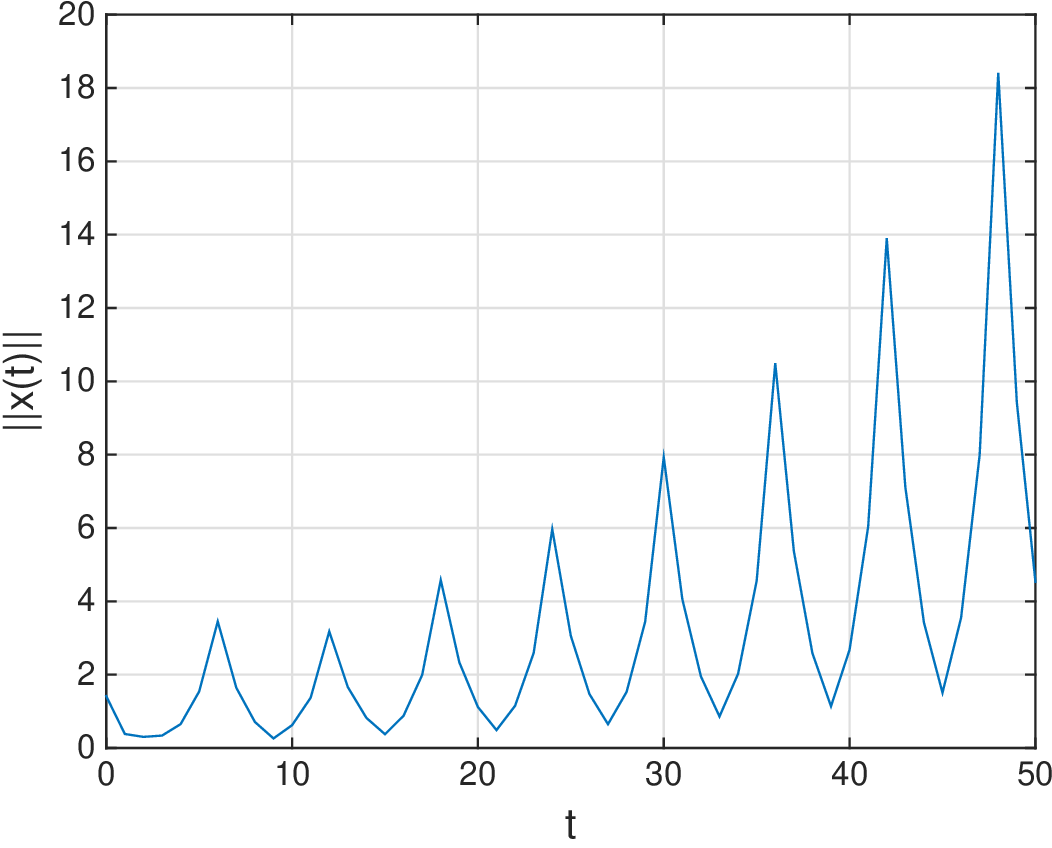}
  		\caption{\(\norm{x(t)}\) versus \(t\)}
  		\label{fig:sub1}
		\end{subfigure}%
		\begin{subfigure}{.5\textwidth}
  		\centering
  			\includegraphics[scale = 0.3]{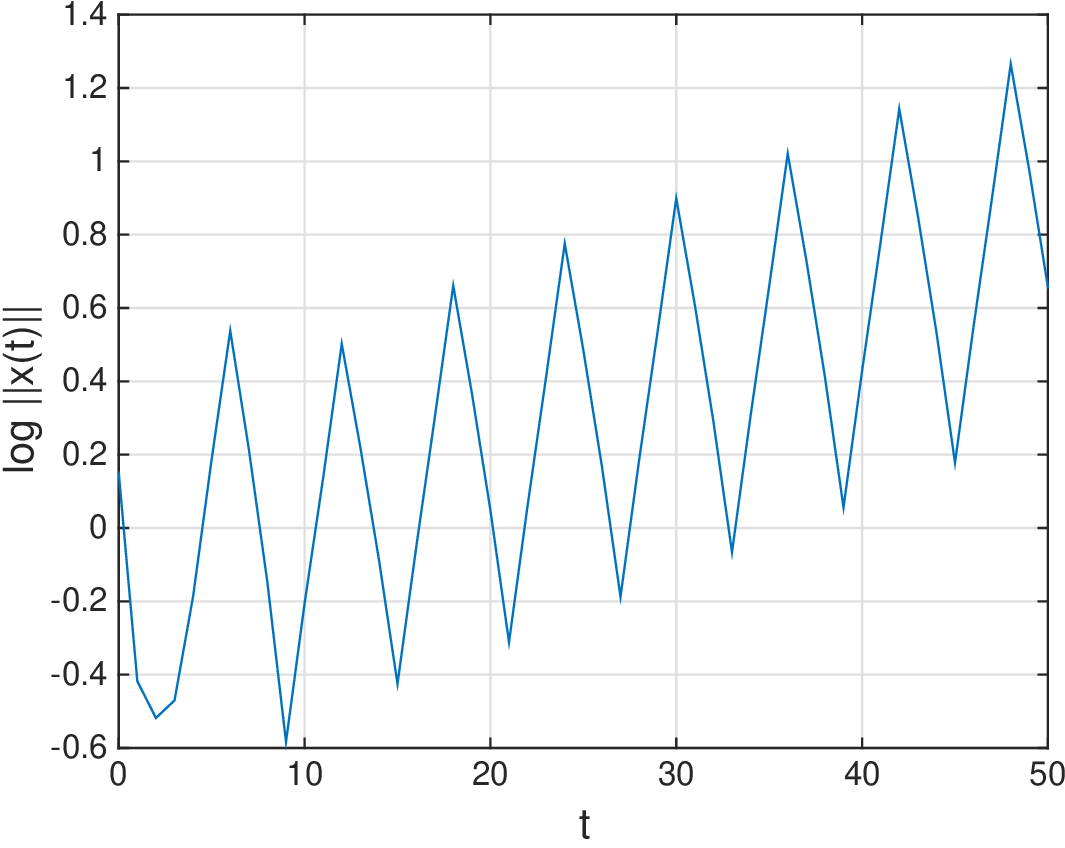}
  		\caption{\(\log \norm{x(t)}\) versus \(t\)}
		\end{subfigure}
		\caption{Plot of \((\norm{x(t)})_{t\in\N_{0}}\) with subsystems \(A_{1}\) and \(A_{2}\) described in Example \ref{ex:counter_ex1}}\label{fig:counter_xplot1}
	\end{figure}
        \begin{figure}[htbp]
	    \centering
		\begin{subfigure}{.5\textwidth}
  		\centering
  			\includegraphics[scale = 0.3]{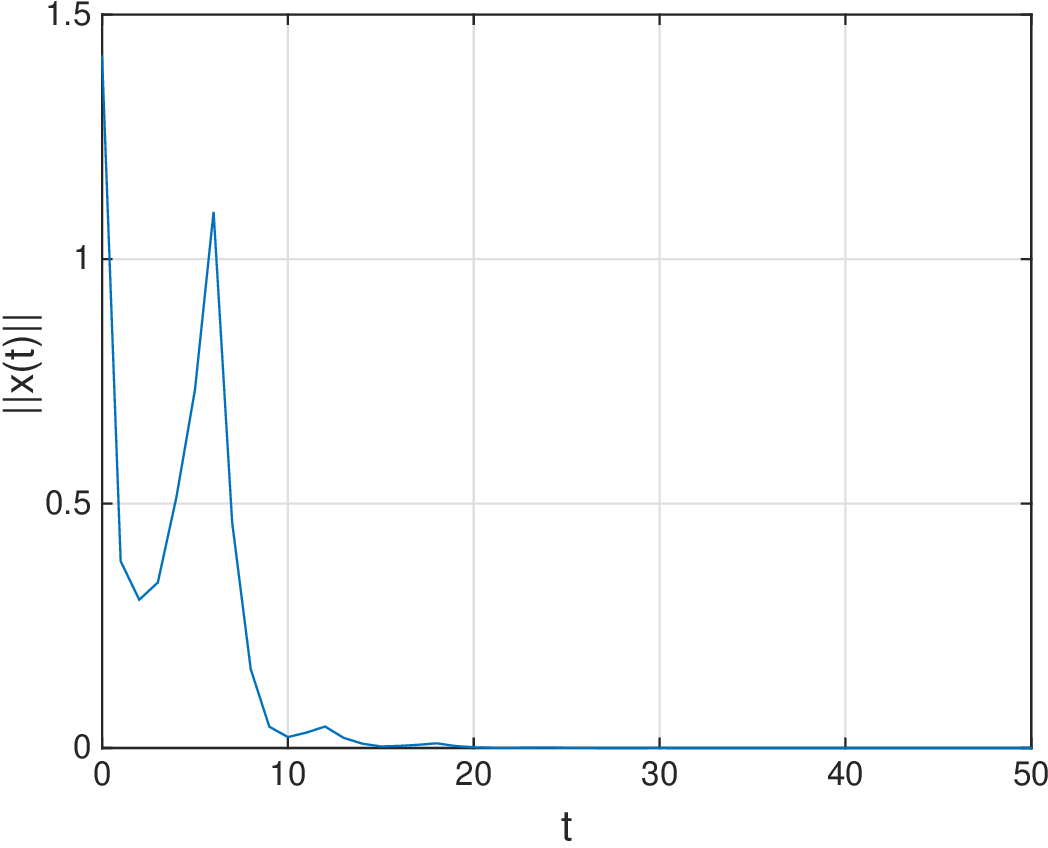}
  		\caption{\(\norm{x(t)}\) versus \(t\)}
  		\label{fig:sub1}
		\end{subfigure}%
		\begin{subfigure}{.5\textwidth}
  		\centering
  			\includegraphics[scale = 0.3]{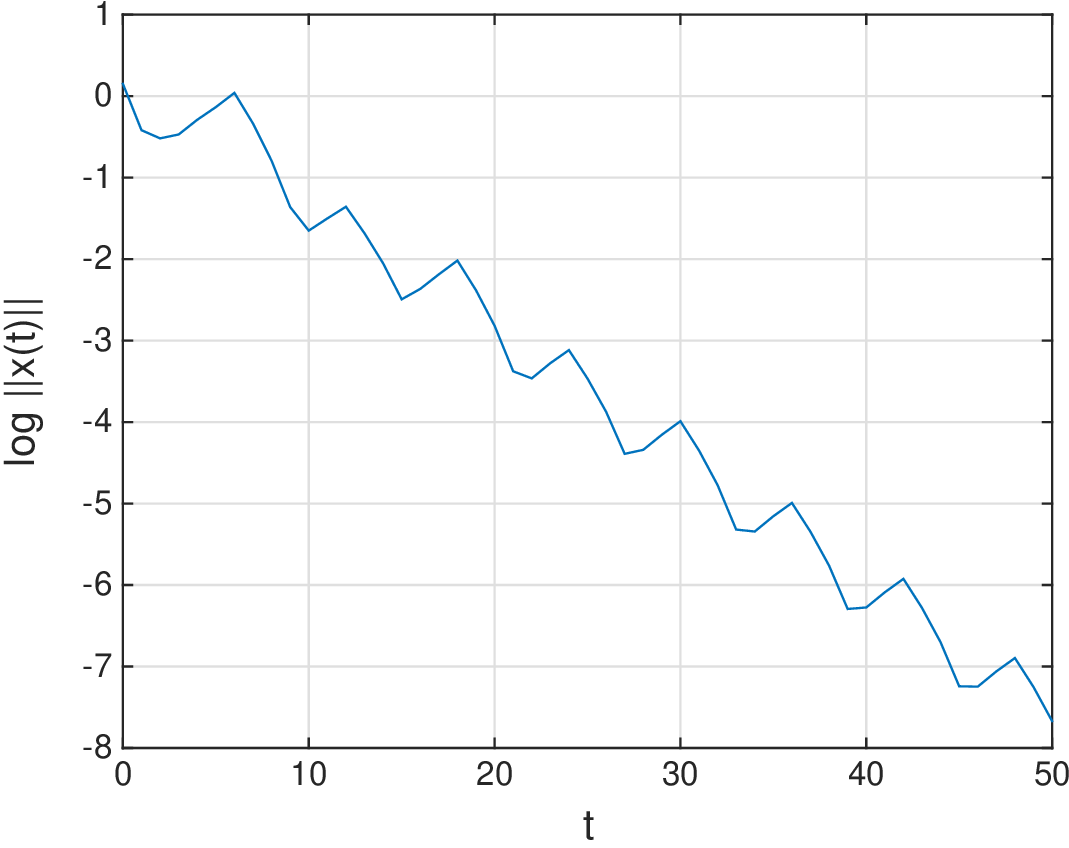}
  		\caption{\(\log \norm{x(t)}\) versus \(t\)}
		\end{subfigure}
		\caption{Plot of \((\norm{x(t)})_{t\in\N_{0}}\) with subsystems \(\tilde{A}_{1}\) and \(\tilde{A}_{2}\) described in Example \ref{ex:counter_ex1}}\label{fig:counter_xplot2}
	   \end{figure}
    \end{example}

    Fix a switching signal \(\sigma\in\tilde{\mathcal{S}}(\Dm,\DM)\). Let \(\tilde{W}\) be the corresponding matrix product defined as: \(\tilde{W} = \ldots A_{\sigma(2)}A_{\sigma(1)}A_{\sigma(0)}\). Let \(\tilde{\mathcal{W}}(\Dm,\DM)\) be the set of all products corresponding to the switching signals belonging to the set \(\tilde{\mathcal{S}}(\Dm,\DM)\). The condition for GUES of \eqref{e:swsys} over the set \(\tilde{\mathcal{S}}(\Dm,\DM)\) can be written equivalently as \cite[\S 2]{Agrachev'12}: for arbitrary choice of \(\tilde{W}\in\tilde{\mathcal{W}}(\Dm,\DM)\), the following condition holds:
    \begin{align}
    \label{e:gues2}
        \norm{\tilde{W}}\leq ce^{-\lambda\abs{\tilde{W}}}\:\text{for all}\:\abs{\tilde{W}}.
    \end{align}

    Our first result identifies conditions on \(\{A_{i}\:|\:i\in\P\}\) such that \eqref{e:gues2} is true.
    \begin{theorem}
    \label{t:mainres1}
        Consider a family of discrete-time linear systems \eqref{e:family}. Let \(\Dm\), \(\DM\in\N\) be given, \(\Dm<\DM\), the matrices \(\{A_{i}\:|\:i\in\P_{S}\}\) satisfy \eqref{e:key_ineq1}, and let \(\lambda\) be an arbitrary positive number satisfying
        \begin{align}
        \label{e:maincondn1}
            \rho e^{\lambda m} < 1.
        \end{align}
        Suppose that the commutators of products of matrices defined in \eqref{e:commutator_defn} satisfy
        \begin{align}
        \label{e:maincondn2}
            E_{ij}^{p,q} = 0_{d\times d}\:\:\text{for all}\:p,q\in\{1,\Dm\}\:\text{and all}\:j\in\P_{S}\:\text{and}\:i\in\P_{U}.
        \end{align}
        Then there exists a positive number \(c\) such that \eqref{e:gues2} holds for arbitrary choice of \(\tilde{W}\in\tilde{\mathcal{W}}(\Dm,\DM)\).
    \end{theorem}

    \begin{rem}
    \label{rem:thm1_des}
        Theorem \ref{t:mainres1} provides a solution to Problem \ref{prob:mainprob}. It relies on commutativity of the matrix products \(A_{i}^{p}\) and \(A_{j}^{q}\) for all \(p,q\in\{1,\Dm\}\) and all \(i\in\P_{U}\) and \(j\in\P_{S}\). Given admissible minimum and maximum dwell times \(\Dm\) and \(\DM\), if the above mentioned matrix products commute, then there exists a positive number \(c\) such that \eqref{e:gues2} holds for arbitrary choice of \(\tilde{W}\in\tilde{\mathcal{W}}(\Dm,\DM)\). In view of Definition \ref{d:gues}, the switched system \eqref{e:swsys} is GUES over the set of switching signals whose elements
        \begin{itemize}[label = \(\circ\), leftmargin = *]
            \item dwell for at least \(m (\geq\Dm)\) and at most \(\DM\) units of time on Schur stable subsystems, and for at least \(\Dm\) and at most \(\DM\) units of time on unstable subsystems, and
            \item do not activate distinct unstable subsystems consecutively.
        \end{itemize}
    \end{rem}

    \begin{rem}
    \label{rem:survey1}
        In \cite{Narendra'94} a switched linear system was shown to be stable under arbitrary switching if the subsystem matrices are Schur stable and commute pairwise. In the setting of arbitrary switching Schur stability of all subsystems is a necessary condition. In contrast, we seek for stability under a subset of the set of all switching signals that satisfy certain pre-specified restrictions on dwell times and accommodate unstable subsystems.
    \end{rem}

	 \begin{rem}
    	\label{rem:robustness_lack}
        		Notice that the stability conditions proposed in Theorem \ref{t:mainres1} are not robust with respect to small perturbations in the elements of the subsystem matrices. Indeed, if the elements of the subsystem matrices \(\{A_{i}\:|\:i\in\P\}\) are perturbed to generate the matrices \(\{\tilde{A}_{i}\:|\:i\in\P\}\), such that the matrix products \(\tilde{A}_{i}^{p}\) and \(\tilde{A}_{j}^{q}\), \(p,q\in\{1,\Dm\}\), \(i\in\P_{U}\), \(j\in\P_{S}\) do not commute, then our stability conditions are no longer useful. However, stability being a robust property, may continue to hold under small perturbations in the elements of the subsystem matrices. This feature motivates our search for stability under a set of subsystems for which the matrix products of our interest do not necessarily commute, but are sufficiently ``close'' to a set of matrices for which these products commute. Our next result characterizes such sets of subsystems.
    \end{rem}

     \begin{theorem}
    \label{t:mainres2}
        Consider a family of discrete-time linear systems \eqref{e:family}. Let \(\Dm\), \(\DM\in\N\) be given, \(\Dm < \DM\), the matrices \(\{A_{i}\:|\:i\in\P_{S}\}\) satisfy \eqref{e:key_ineq1}, and let \(\lambda\) be an arbitrary positive number satisfying \eqref{e:maincondn1}. Suppose that there exist scalars \(\varepsilon_{p,q}\), \(p,q \in \{1,\Dm\}\) small enough such that
        \begin{align}
        \label{e:maincondn11}
            \norm{E_{ij}^{p,q}}\leq\varepsilon_{p,q}\:\text{for all}\:p,q\in\{1,\Dm\}\:\text{and all}\:j\in\P_{S}\:\text{and}\:i\in\P_{U},
        \end{align}
        and
        \begin{align}
        \label{e:maincondn12}
            \rho e^{\lambda m} &+ \Bigl(\zeta_{\Dm,\Dm}(\Dm,\DM)\varepsilon_{\Dm,\Dm}
            +\zeta_{1,\Dm}(\Dm,\DM)\varepsilon_{1,\Dm}
            +\zeta_{\Dm,1}(\Dm,\DM)\varepsilon_{\Dm,1}
            +\zeta_{1,1}(\Dm,\DM)\varepsilon_{1,1}\Bigr)\nonumber\\
            &\quad\quad\times e^{\lambda\bigl(N(m+\DM-1)+1\bigr)}\leq 1.
        \end{align}
        Then there exists a positive number \(c\) such that \eqref{e:gues2} holds for arbitrary choice of \(\tilde{W}\in\tilde{\mathcal{W}}(\Dm,\DM)\).
    \end{theorem}

    \begin{rem}
    \label{rem:thm2_des}
        Theorem \ref{t:mainres2} is our second solution to Problem \ref{prob:mainprob}. We choose a subset of the set of all Schur stable matrices by means of condition \eqref{e:key_ineq1}. If in addition, the Euclidean norms of commutators of products of these matrices \(A_{j}^{p}\) with the products of unstable subsystem matrices \(A_{i}^{q}\), \(p,q\in\{1,\Dm\}\) are bounded above by scalars \(\varepsilon_{p,q}\), \(p,q\in\{1,\Dm\}\) small enough such that condition \eqref{e:maincondn12} holds, then there exists a positive number \(c\) such that \eqref{e:gues2} is true for arbitrary choice of matrix products corresponding to the switching signals \(\sigma\in\tilde{\mathcal{S}}(\Dm,\DM)\). Consequently, the switched system \eqref{e:swsys} is GUES over the set of switching signals \(\tilde{\mathcal{S}}(\Dm,\DM)\).
    \end{rem}

    \begin{rem}
    \label{rem:rob-condn}
    	In contrast to commutativity of matrix products employed in Theorem \ref{t:mainres1}, we utilize a measure of ``closeness'' to commutativity of matrix products in Theorem \ref{t:mainres2}. This technique in spirit is close to \cite{Agrachev'12}. The usage of upper bounds on the norms of \(E_{ij}^{p,q}\), \(p,q\in\{1,\Dm\}\), \(i\in\P_{U}\), \(j\in\P_{S}\) provides inherent robustness to the stability conditions of Theorem \ref{t:mainres2}. Indeed, consider a set of subsystem matrices \(\{A_{i}\:|\:i\in\P\}\), for which Assumption \ref{assump:key1} and conditions \eqref{e:maincondn1}, \eqref{e:maincondn11}-\eqref{e:maincondn12} hold. Now, if the entries of \(\{A_{i}\:|\:i\in\P\}\) are perturbed to generate \(\tilde{A}_{i}\), \(i\in\P\), such that Assumption \ref{assump:key1} and conditions \eqref{e:maincondn1}, \eqref{e:maincondn11}-\eqref{e:maincondn12} continue to hold, then a switched system \eqref{e:swsys} generated by the set of matrices \(\{\tilde{A}_{i}\:|\:i\in\P\}\), continues to be GUES over the set of switching signals \(\tilde{\mathcal{S}}(\Dm,\DM)\).
    \end{rem}

    Prior to discussing other features of our results, we will provide their proofs and explain the analysis technique in detail.

    \begin{proof}[Proof of Theorem \ref{t:mainres1}]
        It suffices to show that if the conditions of Theorem \ref{t:mainres1} hold, then for arbitrary \(\tilde{W}\in\tilde{\mathcal{W}}(\Dm,\DM)\), the condition \eqref{e:gues2} is true. We will employ mathematical induction on \(\abs{\tilde{W}}\) to establish \eqref{e:gues2}.

        \emph{A. Induction basis}: Pick \(c\) large enough so that \eqref{e:gues2} holds for all \(\tilde{W}\) satisfying \(\abs{\tilde{W}}\leq N(m+\DM-1)+1\).

        \emph{B. Induction hypothesis}: Let \(\abs{\tilde{W}}\geq N(m+\DM-1)+2\) and assume that \eqref{e:gues2} holds for all products of length less than \(\abs{\tilde{W}}\).

        \emph{C. Induction step}: Let \(\tilde{W} = LR\), where \(\abs{R} = N(m+\DM-1)+1 = (N-1)(m+\DM-1)+m+\DM\). We observe that there exists an index \(j\in\P_{S}\) such that \(R\) contains at least \(m\) consecutive \(A_{j}\)'s. Indeed, otherwise conditions \eqref{e:sw_res1}-\eqref{e:sw_res2} are violated.

        Without loss of generality, let \(j = 1\) be the first index (reading the product \(\tilde{W}\) from the right) of a Schur stable subsystem, and by the hypothesized properties of a \(\sigma\in\tilde{\mathcal{S}}(\Dm,\DM)\), \(A_{1}\) appears at least for \(m\) consecutive entries. We rewrite \(R\) as
        \[
            R = R_{1}A_{1}^m + R_{2},
        \]
        where \(\abs{R_{1}} = (N-1)(m+\DM-1)\). (Consider, for example, \(N = 2\), \(\P_{S} = \{1\}\), \(\P_{U} = \{2\}\), \(\Dm = 2\), \(\DM = 3\), \(m=3\). Let \(R = \cdots A_{1}A_{1}A_{1}A_{2}A_{2}A_{2}\). It can be rewritten as
        \begin{align*}
            R &= \cdots A_{1}\underline{A_{1}^2A_{2}}A_{2}^2\\
            &= \cdots A_{1}A_{2}\underline{A_{1}^{2}A_{2}^{2}} - \cdots A_{1}E_{21}^{1,2}A_{2}^{2}\\
            &= \cdots \underline{A_{1}A_{2}}A_{2}^{2}A_{1}^{2} - \cdots A_{1}A_{2}E_{21}^{2,2} - \cdots A_{1}E_{21}^{1,2}A_{2}^{2}\\
            &= \cdots A_{2}\underline{A_{1}A_{2}^{2}}A_{1}^{2} - \cdots E_{21}^{1,1}A_{2}^{2}A_{1}^{2} - \cdots A_{1}A_{2}E_{21}^{2,2} - \cdots A_{1}E_{21}^{1,2}A_{2}^{2}\\
            &= \cdots A_{2}A_{2}^{2}A_{1}A_{1}^{2} - \cdots A_{2}E_{21}^{2,1}A_{1}^{2} - \cdots E_{21}^{1,1}A_{2}^{2}A_{1}^{2}- \cdots A_{1}A_{2}E_{21}^{2,2} - \cdots A_{1}E_{21}^{1,2}A_{2}^{2}.)
        \end{align*}
        The sum \(R_{2}\) contains at most
        \begin{itemize}[label = \(\circ\), leftmargin = *]
            \item \(K_{1}K_{2}\) terms of length \((N-1)(m+\DM-1)+m+\DM-2\Dm+1\) with \((N-1)(m+\DM-1)+m+\DM-2\Dm\) \(A_{i}\)'s and \(1\) \(E_{i1}^{\Dm,\Dm}\) (generated by exchanging \(K_{1}\)-many \(A_{1}^{\Dm}\)'s with \(K_{2}\)-many \(A_{i}^{\Dm}\)'s, \(i\in\P_{U}\)),
            \item \(K_{1}(\DM-K_{2}\Dm)\) terms of length \((N-1)(m+\DM-1)+m+\DM-\Dm\) with \((N-1)(m+\DM-1)+m+\DM-\Dm-1\) \(A_{i}\)'s and \(1\) \(E_{i1}^{1,\Dm}\) (generated by exchanging \(K_{1}\)-many \(A_{1}^{\Dm}\)'s with \((\DM-K_{2}\Dm)\)-many \(A_{i}\)'s, \(i\in\P_{U}\)),
            \item \((m-K_{1}\Dm)K_{2}\) terms of length \((N-1)(m+\DM-1)+m+\DM-\Dm\) with \((N-1)(m+\DM-1)+m+\DM-\Dm-1\) \(A_{i}\)'s and \(1\) \(E_{i1}^{\Dm,1}\) (generated by exchanging \((m-K_{1}\Dm)\)-many \(A_{1}\)'s with \(K_{2}\)-many \(A_{i}^{\Dm}\)'s, \(i\in\P_{U}\)), and
            \item \((m-K_{1}\Dm)(\DM-K_{2}\Dm)\) terms of length \((N-1)(m+\DM-1)+m+\DM-1\) with \((N-1)(m+\DM-1)+m+\DM-2\) \(A_{i}\)'s and \(1\) \(E_{i1}^{1,1}\) (generated by exchanging \((m-K_{1}\Dm)\)-many \(A_{1}\)'s with \((\DM-K_{2}\Dm)\)-many \(A_{i}^{\Dm}\)'s, \(i\in\P_{U}\)).
        \end{itemize}

        Now, applying the sub-multiplicativity and sub-additivity properties of the induced Euclidean norm, we obtain
        \begin{align}
        \label{e:pf1_step1}
            &\norm{\tilde{W}} = \norm{LR} \leq \norm{LR_{1}}\norm{A_{1}^{m}}+\norm{L}\norm{R_{2}}\nonumber\\
            &\leq ce^{-\lambda\bigl(\abs{\tilde{W}}-m\bigr)}\rho + ce^{-\lambda\bigl(\abs{\tilde{W}}-(N(m+\DM-1)+1)\bigr)}\times\nonumber\\
            &\hspace*{0.5cm}\Biggl(K_{1}K_{2}\norm{E_{i1}^{\Dm,\Dm}}M^{(N-1)(m+\DM-1)+m+\DM-2\Dm}\nonumber\\
            &\hspace*{0.5cm}+K_{1}(\DM-K_{2}\Dm)\norm{E_{i1}^{1,\Dm}}M^{(N-1)(m+\DM-1)+m+\DM-\Dm-1}\nonumber\\
            &\hspace*{0.5cm}+(m-K_{1}\Dm)K_{2}\norm{E_{i1}^{\Dm,1}}M^{(N-1)(m+\DM-1)+m+\DM-\Dm-1}\nonumber\\
            &\hspace*{0.5cm}+(m-K_{1}\Dm)(\DM-K_{2}\Dm)\norm{E_{i1}^{1,1}}M^{(N-1)(m+\DM-1)+m+\DM-2}\Biggr)\nonumber\\
            &= ce^{-\lambda\bigl(\abs{\tilde{W}}-m\bigr)}\rho + ce^{-\lambda\bigl(\abs{\tilde{W}}-(N(m+\DM-1)+1)\bigr)}\times\nonumber\\
            &\hspace*{0.5cm}\Bigl(\zeta_{\Dm,\Dm}(\Dm,\DM)\norm{E_{i1}^{\Dm,\Dm}}
            +\zeta_{1,\Dm}(\Dm,\DM)\norm{E_{i1}^{1,\Dm}}
            +\zeta_{\Dm,1}(\Dm,\DM)\norm{E_{i1}^{\Dm,1}}
            +\zeta_{1,1}(\Dm,\DM)\norm{E_{i1}^{1,1}}\Bigr),
        \end{align}
        where the upper bounds on \(\norm{LR_{1}}\) and \(\norm{L}\) are obtained by using the relations \(\abs{\tilde{W}_{\Delta}} = \abs{LR_{1}}+\abs{A_{1}^{m}}\) and \(\abs{\tilde{W}_{\Delta}} = \abs{L}+\abs{R}\), respectively. From condition \eqref{e:maincondn2}, we have that
        \begin{align*}
             \zeta_{\Dm,\Dm}(\Dm,\DM)\norm{E_{i1}^{\Dm,\Dm}}
            +\zeta_{1,\Dm}(\Dm,\DM)\norm{E_{i1}^{1,\Dm}}
            +\zeta_{\Dm,1}(\Dm,\DM)\norm{E_{i1}^{\Dm,1}}
            +\zeta_{1,1}(\Dm,\DM)\norm{E_{i1}^{1,1}} = 0.
        \end{align*}
        Consequently, the right-hand side of \eqref{e:pf1_step1} becomes
        \begin{align}
        \label{e:pf1_step2}
            ce^{-\lambda\abs{\tilde{W}}}\cdot\rho e^{\lambda m}.
        \end{align}
        Applying \eqref{e:maincondn1} to \eqref{e:pf1_step2} leads to \eqref{e:gues2}.

        This completes our proof of Theorem \ref{t:mainres1}.
    \end{proof}

    \begin{proof}[Proof of Theorem \ref{t:mainres2}]
        The proof follows with exactly the same set of arguments as in our proof of Theorem \ref{t:mainres1}, except that the commutators of matrix products, \(E_{ij}^{p,q}\), \(p,q\in\{1,\Dm\}\), \(i\in\P_{U}\), \(j\in\P_{S}\) are no longer \(0\) matrices. Recall that we have
        \begin{align*}
            \norm{\tilde{W}} &\leq ce^{-\lambda\bigl(\abs{\tilde{W}}-m\bigr)}\rho + ce^{-\lambda\bigl(\abs{\tilde{W}}-(N(m+\DM-1)+1)\bigr)}\times\nonumber\\
            &\hspace*{0.5cm}\Biggl(\zeta_{\Dm,\Dm}(\Dm,\DM)\norm{E_{i1}^{\Dm,\Dm}}+\zeta_{1,\Dm}(\Dm,\DM)\norm{E_{i1}^{1,\Dm}}
            +\zeta_{\Dm,1}(\Dm,\DM)\norm{E_{i1}^{\Dm,1}}+\zeta_{1,1}(\Dm,\DM)\norm{E_{i1}^{1,1}}\Biggr).
        \end{align*}
        Applying \eqref{e:maincondn11} on the right-hand side of the above inequality, we obtain
        \begin{align}
        \label{e:pf2_step1}
            \norm{\tilde{W}}
            \leq ce^{-\lambda\bigl(\abs{\tilde{W}}-m\bigr)}\rho &+ ce^{-\lambda\bigl(\abs{\tilde{W}}-(N(m+\DM-1)+1)\bigr)}\times\nonumber\\
            &\quad\Bigl(\zeta_{\Dm,\Dm}(\Dm,\DM)\varepsilon_{\Dm,\Dm}
            +\zeta_{1,\Dm}(\Dm,\DM)\varepsilon_{1,\Dm}
            +\zeta_{\Dm,1}(\Dm,\DM)\varepsilon_{\Dm,1}
            +\zeta_{1,1}(\Dm,\DM)\varepsilon_{1,1}\Bigr)\nonumber\\
            =ce^{-\lambda\abs{\tilde{W}}}\Biggl(\rho e^{\lambda m}
            &+\Bigl(\zeta_{\Dm,\Dm}(\Dm,\DM)\varepsilon_{\Dm,\Dm}
            +\zeta_{1,\Dm}(\Dm,\DM)\varepsilon_{1,\Dm}
            +\zeta_{\Dm,1}(\Dm,\DM)\varepsilon_{\Dm,1}
            +\zeta_{1,1}(\Dm,\DM)\varepsilon_{1,1}\Bigr)\nonumber\\
            &\quad\quad\quad\quad\times e^{\lambda\bigl(N(m+\DM-1)+1\bigr)}\Biggr).
        \end{align}
        Applying \eqref{e:maincondn12} to \eqref{e:pf2_step1}, we obtain that \eqref{e:gues2} holds, thereby completing our proof of Theorem \ref{t:mainres2}.
    \end{proof}

        \begin{rem}
    \label{rem:analysis_des}
        The technique of applying counting arguments to matrix products split into sums was applied earlier to cater to arbitrary switching in \cite[Proof of Proposition 1]{Agrachev'12}. The overarching hypothesis there is that all subsystems are Schur stable. In this section we admit unstable systems in the family \eqref{e:family} and focus on a set of switching signals that obeys the given restrictions on dwell times and preserves stability of the switched system \eqref{e:swsys}. The differences of our analysis technique with respect to \cite{Agrachev'12} are highlighted below:
        \begin{itemize}[label = \(\circ\), leftmargin = *]
            \item In \cite{Agrachev'12} the authors split a matrix product \(\overline{W}\) into two sub-products: the left sub-product \(\overline{L}\) and the right sub-product \(\overline{R}\), split \(\overline{L}\) as a sum to arrive at \((\overline{A}_{1}^{n}\overline{L}_{1}+\overline{L}_{2})\overline{R}\), where \(n\in\N\) satisfies \(\norm{\overline{A}_{1}^{n}}\leq\overline{\rho}< 1\), and then apply counting arguments. Here, we split the right sub-product \(R\) of \(\tilde{W}\) into sums by utilizing the structure of the switching signals \(\sigma\in\tilde{\mathcal{S}}(\Dm,\DM)\). In the worst case, the rightmost \(m+\DM\) terms of \(\tilde{W}\), \(\abs{\tilde{W}}\geq m+\DM\), are (reading from the left) \(m\)-many \(A_{j}\)'s followed by \(\DM\)-many \(A_{i}\)'s, \(i\in\P_{U}\), \(j\in\P_{S}\).
            \item The procedure of rearranging \(\overline{W}\) in the form \((\overline{A}_{1}^{n}\overline{L}_{1}+\overline{L}_{2})\overline{R}\) presented in \cite{Agrachev'12} involves exchanging at every step two distinct matrices \(\overline{A}_{i}\) and \(\overline{A}_{j}\) that appear consecutively in \(\overline{L}\), and consequently, the stability conditions involve upper bounds on the norm of the matrix commutators of \(\overline{A}_{i}\) and \(\overline{A}_{j}\). Our procedure to exchange Schur stable and unstable matrices to obtain the form \(L(R_{1}A_{1}^m +R_{2})\) utilizes the structure of a \(\sigma\in\tilde{\mathcal{S}}(\Dm,\DM)\), and involves the following steps:
                \begin{itemize}[label = \(\diamond\),leftmargin = *]
                    \item \(K_{1}\) products of length \(\Dm\) of a Schur stable matrix \(A_{j}\) are exchanged with at most \(K_{2}\) products of length \(\Dm\) and \(\DM-K_{2}\Dm\) entries of an unstable matrix \(A_{i}\), and
                    \item \(m-K_{1}\Dm\) entries of a Schur stable matrix \(A_{j}\) are exchanged with at most \(K_{2}\) products of length \(\Dm\) and \(\DM-K_{2}\Dm\) entries of an unstable matrix \(A_{i}\).
                \end{itemize}
             This leads us to rely on the commutators of matrix products, \(E_{ij}^{1,1}\), \(E_{ij}^{1,\Dm}\), \(E_{ij}^{\Dm,1}\) and \(E_{ij}^{\Dm,\Dm}\) defined in \eqref{e:commutator_defn}. With \(E_{ij}^{p,q} = 0\) for all \(p,q\in\{1,\Dm\}\) and all \(i\in\P_{U}\) and \(j\in\P_{S}\), we arrive at Theorem \ref{t:mainres1}, while to achieve robustness with respect to small perturbations in the elements of the subsystem matrices, we employ \(\norm{E_{ij}^{p,q}} \leq \varepsilon_{p,q}\) for all \(p,q\in\{1,\Dm\}\) and all \(i\in\P_{U}\) and \(j\in\P_{S}\) in Theorem \ref{t:mainres2}.
        \end{itemize}
    \end{rem}

    \begin{rem}
    \label{rem:m_choice}
    	The choice of \(m\) described in Remark \ref{rem:key_assump} is restrictive as far as the elements of \(\tilde{\mathcal{S}}(\Dm,\DM)\) that activate unstable subsystems, are concerned. The restriction is in the sense of the size of the subsets of Schur stable matrices that Theorems \ref{t:mainres1}-\ref{t:mainres2} cater to. Indeed, \(\norm{A_{j}^{k}} < 1\), \(k\in\{m+1,\ldots,\DM\}\) is not utilized explicitly in our proofs of Theorems \ref{t:mainres1}-\ref{t:mainres2}, but are taken care of in condition \eqref{e:maincondn12}. In fact, the use of the smallest \(\Dm\leq m\leq\DM\) satisfying \eqref{e:key_ineq1}, suffices. However, as explained in Remark \ref{rem:all_stable_unstable}, the choice of \(m\) described in Remark \ref{rem:key_assump} is useful for the elements of \(\tilde{\mathcal{S}}(\Dm,\DM)\) that do not activate unstable subsystems at all. To accommodate Schur stable subsystems that satisfy \(\norm{A_{j}^{m}} < 1\) and \(\norm{A_{j}^{n}} > 1\), \(\DM\geq n > m\geq\Dm\), we require rearranging the matrix products \(A_{\sigma(t-1)}\ldots A_{\sigma(1)}A_{\sigma(0)}\) in the form \(A_{j_{1}}^{m}A_{j_{2}}^{m}\ldots\), which leads to the requirement of additional conditions on the commutators of the matrix products \(A_{j_{1}}^{p}\) and \(A_{j_{2}}^{q}\), \(p,q\in\{m,m+1,\ldots,\DM\}\), \(j_{1},j_{2}\in\P_{S}\). The reader is referred to \cite[Remark 7]{mno} for a discussion on matrix commutator based stability conditions that cater to restricted switching with all Schur stable subsystems.
    \end{rem}

    \begin{rem}
    \label{rem:Lyap_condn}
    	A commonly used tool for studying stability of switched systems under dwell time switching is multiple Lyapunov-like functions \cite{Branicky'98}. The analysis technique involves compensating the maximum increase in these functions caused by activation of unstable subsystems and occurrence of switches, by the minimum decrease in these functions caused by activation of stable subsystems, see e.g., \cite{abc,def,ghi,Liberzon'12}. In this article we follow a paradigm shift, and rely on commutation relations between subsystem matrices. Matrix commutators (Lie brackets) have been used widely to cater to arbitrary switching earlier in the literature. A switched linear system is stable under arbitrary switching if the subsystem matrices pairwise commute \cite{Narendra'94}, Lie algebra is nilpotent \cite{Gurvitz'95}, solvable \cite{Liberzon'99}, or has a compact semisimple part \cite{Agrachev'01}. In \cite{Haimovich'13} the authors addressed the problem of designing state-feedback matrices such that Lie algebra associated to the closed-loop subsystems is approximately solvable. Robustness of matrix commutation relations with respect to small perturbations in the elements of the subsystem matrices was addressed in \cite{Liberzon'09} for periodic switching and in \cite{Agrachev'12} for arbitrary switching. Here, we extend the analysis technique of \cite{Agrachev'12} to the setting of restricted switching in the presence of unstable subsystems.  The use of matrix commutators allows us to characterize stability directly in terms of properties of subsystem matrices, and not in terms of existence of certain classes of Lyapunov-like functions.
    \end{rem}

    \begin{rem}
    \label{rem:compa}
        In \cite{abc,def} the design of stabilizing switching signals under dwell time constraints involves constructing negative weight cycles on the underlying weighted digraph of a switched system. While the techniques of \cite{abc,def} cater to general nonlinear setting, the existence of stabilizing cycles depends on the existence of Lyapunov-like functions that satisfy certain conditions individually and among themselves. Given a family of systems, designing such functions is, in general, a numerically difficult problem. In contrast, the results proposed in this article do not involve verifying if suitable Lyapunov-like functions exist for a given family of systems, instead checking certain properties of the subsystem matrices is sufficient. Our stability conditions are, however, limited to the case of switched linear systems unlike Lyapunov-like function based techniques that extend to switched nonlinear systems under standard assumptions.
    \end{rem}
\section{Numerical examples}
\label{s:num_ex}
	\begin{example}
	\label{ex:exmpl1}
		We consider \(\P = \{1,2\}\) with
		\[
			A_{1} = \pmat{-0.92 & 0\\0 & 0.77}\:\:\text{and}\:\:A_{2} = \pmat{1.24 & 0\\0 & 0.89}.
		\]
		Clearly, \(\P_{S} = \{1\}\) and \(\P_{U} = \{2\}\). Let \(\Dm = 2\) and \(\DM = 3\). We have
		\[
			\norm{A_{1}^{2}} = 0.85,\:\:\norm{A_{1}^{3}} = 0.78.
		\]
		Hence, \(m = 2\) and \(\rho = 0.85\). Let \(\lambda = 0.001\), which leads to
		\[
			\rho e^{\lambda m} = 0.85 < 1.
		\]
		Also,
		\[
			\norm{E_{21}^{2,2}} = 0,\:\:\norm{E_{21}^{1,2}} = 0,\:\:\norm{E_{21}^{2,1}} = 0,\:\:\norm{E_{21}^{1,1}} = 0.
		\]
		Consequently, the conditions of Theorem \ref{t:mainres1} hold.
		
		We generate \(1000\) random switching signals that obey conditions \eqref{e:sw_res1}-\eqref{e:sw_res2} and plot the corresponding \((\norm{x(t)})_{t\in\N_{0}}\) in Figure \ref{fig:ex1_xplots}. The initial conditions \(x_{0}\) are chosen uniformly at random from the interval \([-100,100]^{2}\). We observe that the switched system \eqref{e:swsys} is GUES under all these signals.
		\begin{figure}[htbp]
	    \centering
		\begin{subfigure}{.5\textwidth}
  		\centering
  			\includegraphics[scale = 0.3]{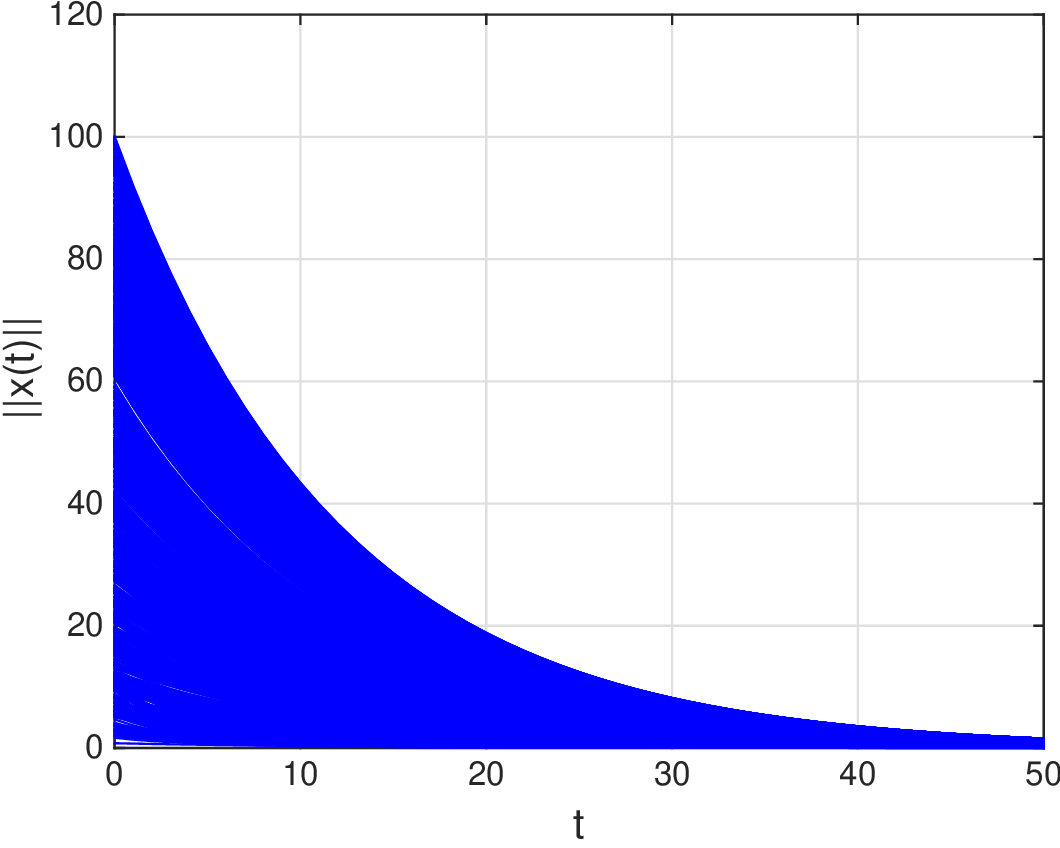}
  		\caption{\(\norm{x(t)}\) versus \(t\)}
  		\label{fig:sub1}
		\end{subfigure}%
		\begin{subfigure}{.5\textwidth}
  		\centering
  			\includegraphics[scale = 0.3]{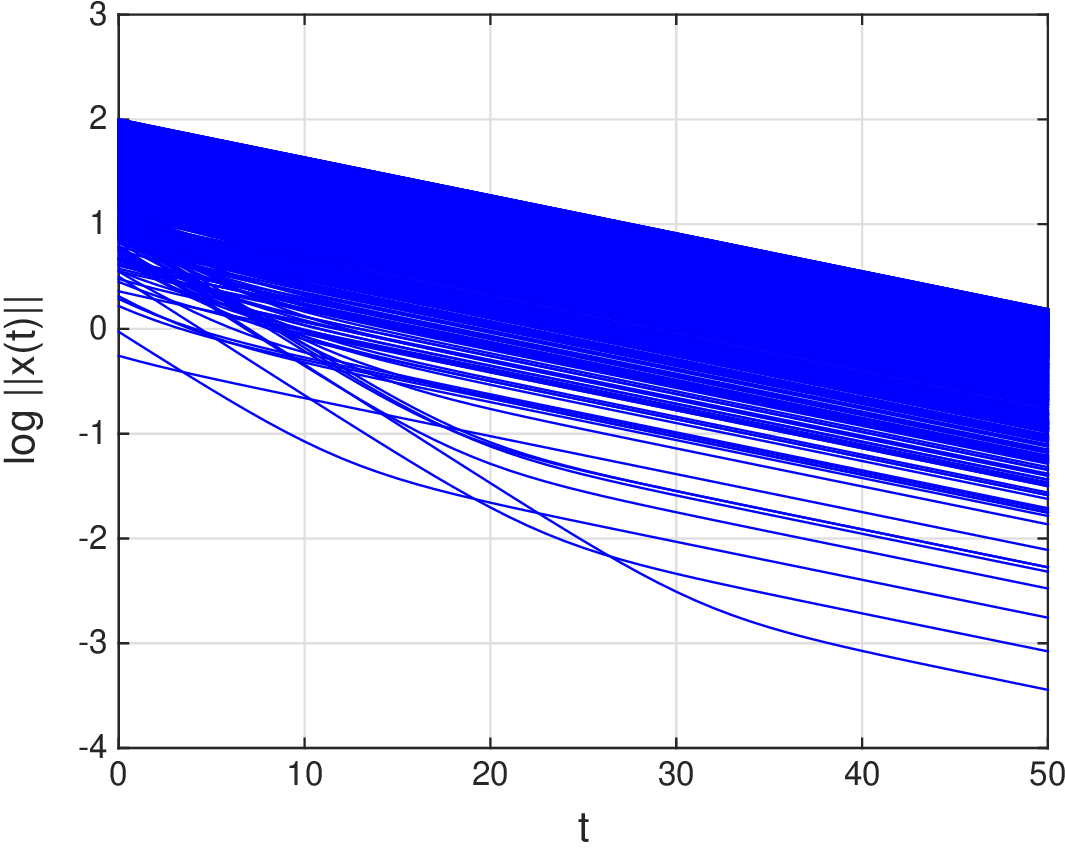}
  		\caption{\(\log \norm{x(t)}\) versus \(t\)}
		\end{subfigure}
		\caption{Plot of \((\norm{x(t)})_{t\in\N_{0}}\) for Example \ref{ex:exmpl1}}\label{fig:ex1_xplots}
	   \end{figure}
	\end{example}
	
	\begin{example}
	\label{ex:exmpl2}
		Consider the setting of Example \ref{ex:exmpl1}. We now perturb the elements of \(\{A_{i}\:|\:i\in\P\}\) to generate
		\begin{align*}
			\tilde{A}_{1} &= A_{1} + \pmat{0 & 0.1\\0 & 0} = \pmat{-0.92 & 0.1\\0 & 0.77}\:\:\text{and}\\
			\tilde{A}_{2} &= A_{2} +  \pmat{0 & 0\\0.05 & 0} = \pmat{1.24 & 0\\0.05 & 0.89}.
		\end{align*}
		The above perturbations preserve \(\P_{S} = \{1\}\) and \(\P_{U} = \{2\}\). However, the matrices \(\tilde{A}_{2}^{p}\) and \(\tilde{A}_{1}^{q}\), \(p,q\in\{1,\Dm\}\) no longer commute. Indeed,
		\[
			\norm{E_{21}^{2,2}} = 0.0272,\:\:\norm{E_{21}^{1,2}} = 0.0127,\:\:\norm{E_{21}^{2,1}} = 0.1811,\:\:\norm{E_{21}^{1,1}} = 0.0850.
		\]
		We will apply Theorem \ref{t:mainres2}. We have
		\begin{align*}
			&M = \max\{\norm{A_{1}},\norm{A_{2}}\} = 1.24,\\
			&\norm{A_{1}^{2}} = 0.85,\:\:\norm{A_{1}^{3}} = 0.78,\\
			&\rho = 0.85,\:\:m = 2,\:\:\lambda = 0.001,\\
			&K_{1} = \lfloor\frac{m}{\Dm}\rfloor = 1,\:\:K_{2} = \lfloor\frac{\DM}{\Dm}\rfloor = 1,\\
			&\zeta_{2,2}(2,3) = 2.93,\:\:\zeta_{1,2}(2,3) = 3.64,
			\zeta_{2,1}(2,3) = 0,\:\:\zeta_{1,1}(2,3) = 0,\\
			&e^{\lambda\bigl(N(m+\DM-1)+1\bigr)} = 1.0090.
		\end{align*}
		
		Consequently,
		\begin{align*}
			&\rho e^{\lambda m} + \Bigl(\zeta_{\Dm,\Dm}(\Dm,\DM)\varepsilon_{\Dm,\Dm}
            		+\zeta_{1,\Dm}(\Dm,\DM)\varepsilon_{1,\Dm}
            		+\zeta_{\Dm,1}(\Dm,\DM)\varepsilon_{\Dm,1}
            		+\zeta_{1,1}(\Dm,\DM)\varepsilon_{1,1}\Bigr)\nonumber\\
            		&\hspace*{2cm}\quad\quad\times e^{\lambda\bigl(N(m+\DM-1)+1\bigr)}\\
			=&\:\: 0.85\times e^{2\times 0.001} + (2.93\times 0.0272 + 3.64\times 0.0127 + 0 + 0)\times 1.0090\\
			=&\:\: 0.98 < 1,
            	\end{align*}
		and the conditions of Theorem \ref{t:mainres2} hold.
		
		We generate \(1000\) random switching signals that obey conditions \eqref{e:sw_res1}-\eqref{e:sw_res2} and plot the corresponding \((\norm{x(t)})_{t\in\N_{0}}\) in Figure \ref{fig:ex2_xplots}. The initial conditions \(x_{0}\) are chosen uniformly at random from the interval \([-100,100]^{2}\). We observe that the switched system \eqref{e:swsys} is GUES under all these signals.
		\begin{figure}[htbp]
	    \centering
		\begin{subfigure}{.5\textwidth}
  		\centering
  			\includegraphics[scale = 0.3]{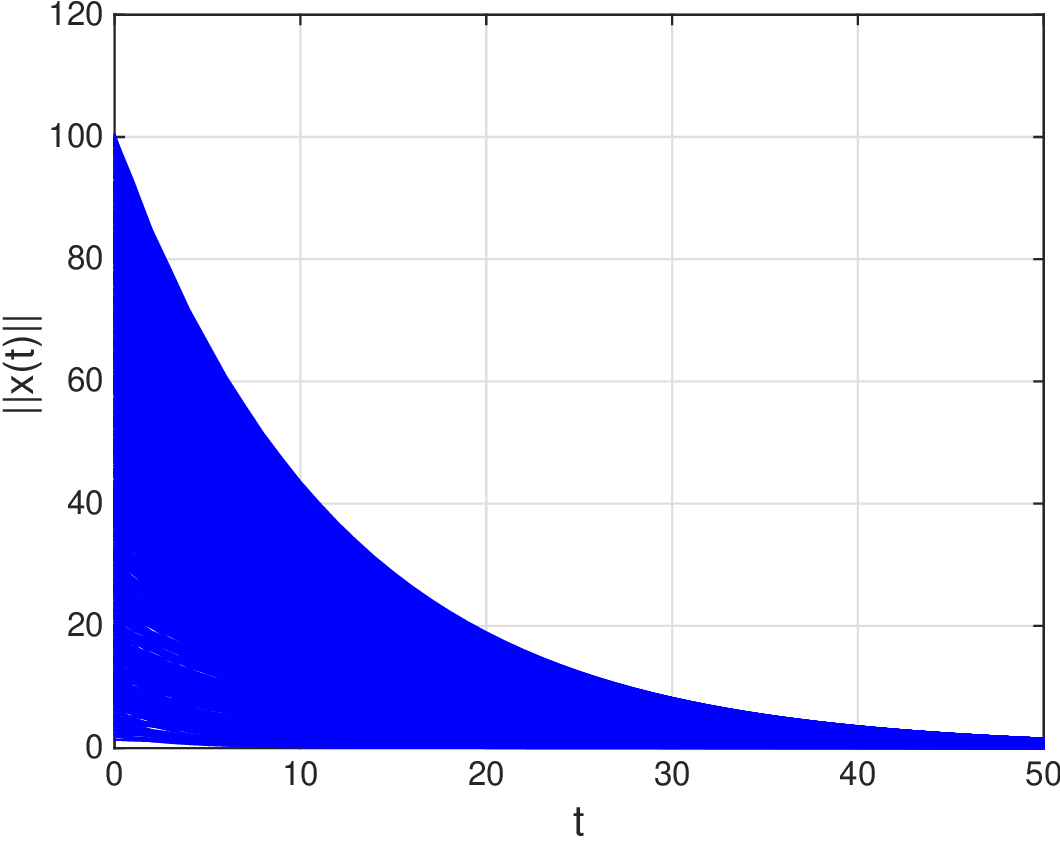}
  		\caption{\(\norm{x(t)}\) versus \(t\)}
  		\label{fig:sub1}
		\end{subfigure}%
		\begin{subfigure}{.5\textwidth}
  		\centering
  			\includegraphics[scale = 0.3]{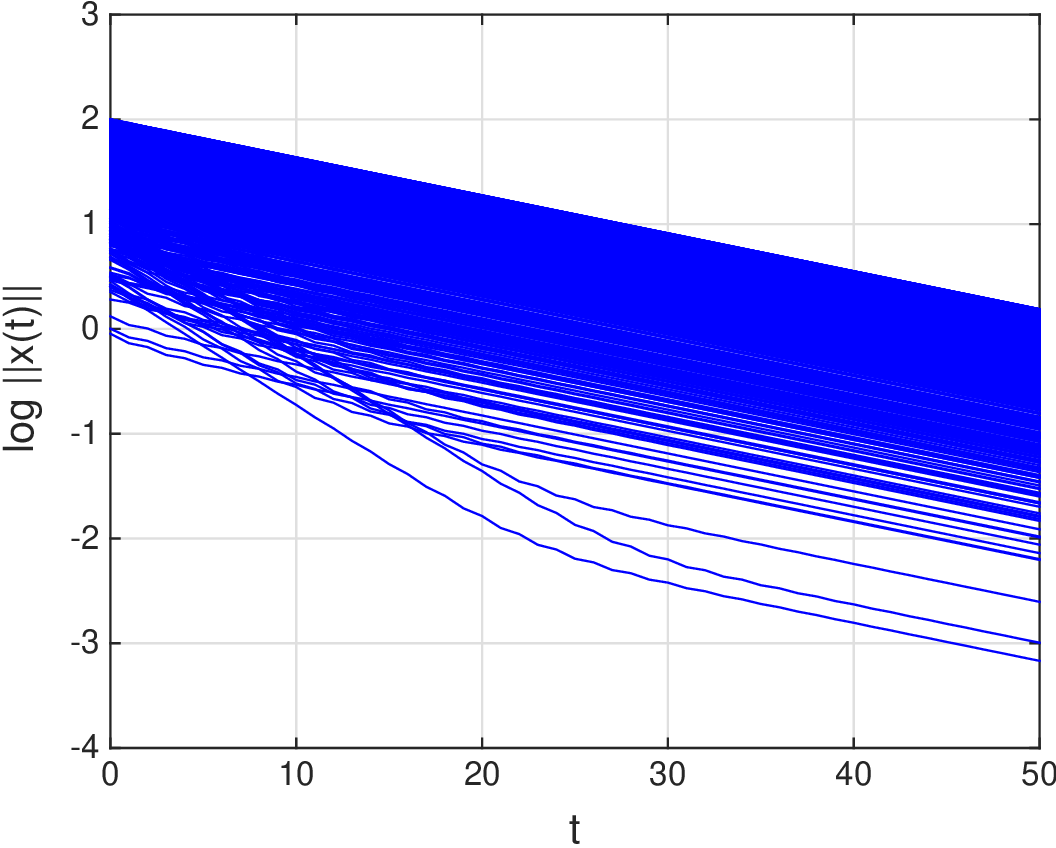}
  		\caption{\(\log \norm{x(t)}\) versus \(t\)}
		\end{subfigure}
		\caption{Plot of \((\norm{x(t)})_{t\in\N_{0}}\) for Example \ref{ex:exmpl1}}\label{fig:ex2_xplots}
	   \end{figure}
	\end{example}
	
\section{Concluding remarks}
\label{s:concln}
     To summarize, we identified sufficient conditions on the subsystems of a switched system such that they admit a set of switching signals that obeys pre-specified restrictions on admissible minimum and maximum dwell times and preserves stability of the resulting switched system. Our set of stabilizing switching signals is characterized in terms of dwell times on Schur stable subsystems and non-consecutive activation of distinct unstable subsystems.

    In the recent past stabilizing switching signals for discrete-time switched linear systems that rely solely on the asymptotic behaviour of these signals was proposed in \cite{ghi}. The characterization of these switching signals depends on the existence of a family of Lyapunov-like functions, the elements of which satisfy certain conditions individually and among themselves. Recently in \cite{Balachandran'19} the authors characterized sets of subsystems that admit the set of stabilizing switching signals proposed in \cite{ghi}. However, the said characterization is in terms of existence of Lyapunov-like functions, and the design of these functions from the subsystem matrices is not addressed. Identifying conditions on subsystems such that they admit large classes of stabilizing switching signals (e.g., the one proposed in \cite{ghi}) directly in terms of properties of the subsystem matrices is an open problem. We envision that the combinatorial techniques presented in this article is a potential tool to address this setting.


\bigskip

\end{document}